\documentclass[preprint, 12pt]{imsart}
\usepackage{graphicx}
\usepackage{epsfig}
\usepackage{amsmath}
\usepackage{amsthm}
\usepackage[round]{natbib}
\usepackage{color}

\usepackage{units}

\startlocaldefs
\newcommand{\n}[1] {\mbox{\boldmath{$#1$}}}

\newtheorem{lem}{Lemma}
\newtheorem{thm}{Theorem}
\newcommand{\be}{\begin{eqnarray}}
\newcommand{\ee}{\end{eqnarray}}
\newcommand{\beq}[1]{\begin{equation}\label{#1}}
\newcommand{\eeq}{\end{equation}}
\newcommand{\ba}{\begin{eqnarray*}}
\newcommand{\ea}{\end{eqnarray*}}

\DeclareMathOperator{\rank}{rank}

\newcommand{\erre}{\mbox{I$\!$R}}
\endlocaldefs

\begin{document}

\begin{frontmatter}

\title{Handling Factors in Variable Selection Problems}
\runtitle{Factors in Variable Selection}

\begin{aug}
  \author{\fnms{Gonzalo} \snm{Garcia-Donato}\ead[label=e1]{gonzalo.garciadonato@uclm.es}}
\and
  \author{\fnms{Rui} \snm{Paulo}\thanksref{t1}\ead[label=e2]{rui@iseg.ulisboa.pt}}

  \thankstext{t1}{Partially supported by the Project CEMAPRE - UID/MULTI/00491/2013 financed by FCT/MCTES through national funds.}
%
  \runauthor{Garcia-Donato and Paulo}

  \affiliation{Universidad de Castilla-La Mancha and Universidade de Lisboa}

  \address{Department of Economics and Finance\\ Universidad de Castilla-La Mancha\\
  Instituto de Desarrollo Regional\\ \printead{e1}}
 
  \address{CEMAPRE and Department of Mathematics\\Lisbon School of Economics Management\\Universidade de Lisboa\\
  \printead{e2}}

\end{aug}

\begin{abstract}
Factors are categorical variables, and the values which these variables assume are called levels. In this paper, we consider the variable selection problem where the set of potential predictors contains both factors and numerical variables. Formally, this problem is a particular case of the standard variable selection problem where factors are coded using dummy variables. As such, the Bayesian solution would be straightforward and, possibly because of this, the problem, despite its importance, has not received much attention in the literature. Nevertheless, we show that this perception is illusory and that in fact several inputs like the assignment of prior probabilities over the model space or the parameterization adopted for factors may have a large (and difficult to anticipate) impact on the results. We provide a solution to these issues that extends the proposals in the standard variable selection problem and does not depend on how the factors are coded using dummy variables. Our approach is illustrated with a real example concerning a childhood obesity study in Spain.
\end{abstract}

\end{frontmatter}
\section{Introduction}
Variable selection in the context of Gaussian regression models has always been a very important topic of research in Statistics, 
and in particular
within the Bayesian community. The issues addressed in the Bayesian literature include computational challenges that stem from high-dimensional problems, the specification of default priors for model-specific parameters and priors for the model space, but also the study of frequentist properties of the resulting methodologies, specifically questions related to consistency.  

Throughout most of this literature, and in particular in the papers dealing with prior specification, the explanatory variables are assumed to be numeric, therefore excluding categorical predictors. Categorical predictors are often called factors, and the different categories that they assume are referred to as levels. The use of dummy variables allows one to formally write the ensuing model as a linear model, so that in principle one could expect that the general recommendations would be readily applicable. In this paper, we investigate issues that arise when one wants to include factors in the list of possible predictors and chooses to follow an objective Bayes approach to variable selection. 


Bayesian variable selection ideas have been utilized with success in the presence of factors to develop experimental designs, both screening and follow-up. Recent examples include \cite{BinChi07} and \cite{ConDel16}. Typically the problem consists in finding adequate designs to ascertain which factors are active in explaining an outcome variable. The assumed model is Gaussian and linear, and the criterion is a function of posterior model probabilities and of a distance between the predictive densities under competing models. These papers are concerned with main effects and interactions, but consider only full-rank models. In the present article, we address the issue of including rank-deficient models and the role of reparametrizations. \cite{BinChi07} use independent priors for all regression coefficients, which require some tuning, whereas \cite{ConDel16} utilizes the priors in \cite{Baetal11}. In common, these papers have the careful specification of the prior in the model space, which is also an important part of our paper. To the best of our knowledge, the question of which levels of a particular factor are most relevant is never considered in the literature of experimental design, and this is a research question to which we provide an answer. 

An approach to the analysis of factorial experiments which is also Bayesian is that of \cite{NobGre00}. Here, finite mixtures are used to represent main effects and interactions. If the main effect of two factors use the same components of the mixture, then they have the same effect on the response. This introduces the notion of a partition of the set of main effects and interactions. By utilizing a reversible-jump Markov chain Monte Carlo algorithm, one is able to determine the most probable partition patterns of the main effects and interactions, which may be seen as an analog of variable selection in the sense that it answers, e.g., the question of which factors have the same main effect. Identifiability constraints are imposed to ensure that the entertained models are full-rank. The priors require considerable tuning and induce in a non-trivial way a prior distribution on the set of partitions, which plays a role similar to the prior on the model space in the approach that we recommend in this paper.

\cite{Rouetal12} and \cite{Rouetal17} specifically address the issue of constructing default Bayes factors for models that include factors as explanatory variables, namely ANOVA and factorial designs. Emphasis is placed on parametrizations and appropriate independent priors. The fact that our proposal does not rely on any reparametrization of the models is a distinguishing aspect of our work when compared to theirs. Any discussion concerning the prior on the model space is completely absent in these two papers.

\cite{ClyPar98} is also a relevant reference. They consider the problem of variable selection in the context of an application where continuous and categorical predictors are present. Again, independent priors are placed on regression coefficients, and a constant prior is placed on the model space, but the idea that a factor is relevant if at least one of its levels is present in the model is considered, as it is in our work.

\cite{Chi96} focuses on constructing prior distributions on the model space that incorporate certain types of relations between predictors, which includes the case of the dummy variables often used to incorporate factors in linear models. Our work and \cite{Chi96} differ in many aspects (later detailed), but \cite{Chi96} observed and reported a number of key concerns (e.g. multiplicity issues and the role of reparametrizations) that are also central in our paper and that we specifically address. In addition, we also tackle the problem of specifying prior distributions on the model-specific parameters and explore the value of posterior inclusion probabilities in this context.

Another approach to dealing with categorical predictors focuses on modeling and fitting techniques which explicitly consider the special nature of these predictors, with a particular emphasis on regularization-based methods. \cite{TutGer16} and \cite{PauWag17} are examples of such approaches, including the references therein. Of particular interest is the removal of a factor or the fusion of levels of a factor. These approaches cannot be viewed as model selection techniques, as they typically do not explicitly entertain the notion of competing models.

The rest of the paper is organized as follows: in Section~\ref{example} we describe the problem, clarify the language that we use and introduce an example that will be used throughout the paper. In Section~\ref{sec.vsrm} we review basic results of variable selection in regression models and lay the ground for the developments that will be described in the next section. Indeed, Section~\ref{sec.factors} is devoted to the question of how to handle factors in variable selection problems. For ease of presentation, the main ideas are first introduced in the one-factor case, and later extended to the general scenario. The paper concluded with a discussion. All proofs are relegated to the appendix.

\section{The problem, nomenclature and an illustrative real example}\label{example}
Before we proceed, let us fix the nomenclature that will be used in the sequel. We refer to any explanatory variable that we entertain as having an effect on certain response $y$ as \emph{predictor\/}. It will be assumed that $y$ is Gaussian, while the effect of the predictors is 
linear.
A numerical predictor will be simply called a \emph{variable\/}. In the variable selection problem we are interested in understanding which variables from a set $\{ x_1,\ldots, x_k \}$ are relevant to explain $y$, while there may be another set of variables $\{x_{01},\ldots,x_{0 k_0}\}$ which are known to affect $y$. We refer to these last ones as \emph{sure\/} variables, and we assume throughout the paper that the constant is a sure variable, i.e., that there is always an intercept in all the models that we consider. 

As we stated in the Introduction, we use the term \emph{factor\/} to refer to a categorical predictor. Our paper discusses and proposes methodology to deal with the model selection problem where, in addition to variables, one considers $p$ factors, $\{A_1,\ldots,A_p\}$, and is interested in ascertaining their role in explaining the response $y$. The number of levels that the factor $A_r$ can take is denoted as $\ell_r$. 

Throughout the paper, and mainly for illustrative purposes, we consider a real example that studies obesity in children. 
This study has been conducted in
Spain \citep{Zuetal11}, and as part of it the body mass index $y$ and several other relevant sources of information concerning children between 2 and 14 years of age were collected. Here, we will consider the predictors in Table~\ref{Obe.Tab}. Apart from the intercept, the weight and height at birth ($x_{01}$ and $x_{02}$, respectively), and the age of the child ($x_{03}$) will be treated as sure variables. Hence, we have a total of $k_0=4$ sure variables. As potential predictors we have a total of five. Among these, there are two variables (hours per day devoted to screens and number of hours devoted daily to sleeping, denoted $x_{1}$ and $x_{2}$, respectively) and two factors: sports activity ($A_1$) and healthiness of food ($A_2$). Factor $A_1$ has $\ell_1=6$ levels, while factor $A_2$ has $\ell_2=3$, ranging from less to more (\emph{a priori}) beneficial habits.

Of the data collected, we only use the set of children for which all predictors have been recorded (without missing values) resulting in a 
total of $1002$ observations.

\begin{table}[t!]
\begin{center}
{\small\scalebox{0.75}{
\begin{tabular}{r|r|c|c}
Predictors & Type & Key & Numer of levels ($\ell$)\\
\hline
Weight at birth &  Sure variable & $x_{01}$ & -\\
Height at birth &  Sure variable & $x_{02}$& -\\
Age & Sure variable & $x_{03}$& -\\
Hours per day devoted to screens (TV, ps3, etc)& Variable & $x_1$& - \\
Hours he/she sleeps & Variable & $x_2$ & - \\
Sports & Factor & $A_1$ & $\ell_1=6$\\
Healthy food & Factor & $A_2$ & $\ell_2=3$
\end{tabular}
}}
\end{center}
\caption{\small Description of the predictors for body mass index, $y$, considered in the obesity example. The constant is also considered a sure variable.}
\label{Obe.Tab}
\end{table}

\section{Variable Selection in Regression Models}\label{sec.vsrm}
The variable selection problem has received considerable attention from the Bayesian community. In this setting, and writing $\n y=(y_1, \ldots, y_n)$, the ensuing statistical model that contains all possible variables (usually called the full model) is
$$
\n y\mid \n \alpha,\n \beta,\sigma\sim N(\n X_0 \n \alpha+\n X \n \beta,\sigma^2 \n I_n)
$$
where $\n I_n$ is the order $n$ identity matrix, $\n \alpha$ and $\n \beta$ are the regression coefficients, and $\sigma^2$ is the variance of the error term. The matrix $\n X_0$ is obtained by collecting the values of the sure variables for each individual $i$ by rows, so that $\n X_0$ is $n\times k_0$ (recall that this matrix contains a vector of ones). Similarly $\n X$ is $n\times k$ and contains the values of all the entertained variables. Throughout this paper, we assume that we have more data points than predictors (i.e., that $n\geq k_0 + k + 1$); see \citet{Beretal16} for a treatment of the problem when such restriction is not met.

The variable selection problem can then be formulated as quantifying the evidence provided by the data in favor of each of the models
\begin{equation}\label{Mgamma}
M_{\boldsymbol \gamma}:\  \n y\mid \n \alpha,\n \beta,\sigma\sim N(\n X_0 \n \alpha+\n X_{{\boldsymbol \gamma}} \n \beta_{\boldsymbol \gamma},\sigma^2 \n I_n)\ ,
\end{equation}
where ${\boldsymbol \gamma}\in\{0,1\}^k$ indicates which of the $k$ variables are present in the model, $\n X_{\boldsymbol \gamma}$ results from selecting the corresponding columns in $\n X$, and similarly for $\n \beta_{\boldsymbol \gamma}$. Slightly abusing notation, the model with none of the variables (the so-called null model) corresponds to ${\boldsymbol \gamma}=\n 0$; the full model is obtained with ${\boldsymbol \gamma}=\n 1$. In the sequel, we denote by $k_{\boldsymbol \gamma}$ the number of variables included under model $M_{\boldsymbol \gamma}$, that is, $\n 1^T \boldsymbol \gamma$.

The Bayesian answer that we adhere to in this article is based on the posterior probabilities of each of the competing $2^k$ models,
\begin{equation}\label{postprob}
P(M_{\boldsymbol \gamma}\mid \n y)\propto m_{{\boldsymbol \gamma}}(\n y)\ P(M_{\boldsymbol \gamma})
\end{equation}
where $P(M_{\boldsymbol \gamma})$ represents the prior probability of model $M_{\boldsymbol \gamma}$, and $m_{\boldsymbol \gamma}(\n y)$ is the prior predictive density of the data under model $M_{\boldsymbol \gamma}$,
$$
m_{{\boldsymbol \gamma}}(\n y)=\int N(\n y\mid \n X_0 \n \alpha+\n X_{{\boldsymbol \gamma}}\n\beta_{\boldsymbol \gamma},\sigma^2 \n I_n)\ \pi_{\boldsymbol \gamma}(\n \alpha, \n \beta_{\boldsymbol \gamma}, \sigma)\ d\n \alpha\ d \n\beta_{\boldsymbol \gamma}\ d\sigma\ ,
$$
with $\pi_{\boldsymbol \gamma}(\n \alpha, \n \beta_{\boldsymbol \gamma}, \sigma)$ denoting the prior distribution on the model-specific parameters. Alternatively, we can rewrite \eqref{postprob} as
$$
P(M_{\boldsymbol \gamma}\mid \n y) = \frac{B_{\boldsymbol \gamma}\ P(M_{\boldsymbol \gamma})}{\sum_{\boldsymbol \gamma'} B_{\boldsymbol \gamma'}\ P(M_{\boldsymbol \gamma'})}
$$
where $B_{\boldsymbol \gamma}  = m_{\boldsymbol \gamma}(\n y)/m_{\boldsymbol 0}(\n y)$ is the so-called Bayes factor of model $M_{\boldsymbol \gamma}$ to the null model.

Standard objective variable selection choices for $P(M_{\boldsymbol \gamma})$ 
include the constant prior
\begin{equation}\label{eq.const}
P(M_{\boldsymbol \gamma})=1/2^{k}\ ,
\end{equation}
which is frequently utilized as it, at least apparently, is the natural non-informative choice. We much prefer the \cite{ScottBerger09} prior that automatically accounts for multiplicity and will be ultimately part of our proposal:
\begin{equation}\label{eq.SB}
P(M_{\boldsymbol \gamma})=\frac{1}{(k+1){k\choose k_{\boldsymbol \gamma}}}\ .
\end{equation}
This prior is a particular case of the more general beta-binomial prior, which has been used to incorporate prior knowledge about the true model size by e.g.\ \cite{LeySteel12}. Considerations about special characteristics of the underlying problem (e.g.\ collinearity issues) have lead to other interesting alternatives, including the dilution priors of \cite{George10} and the model space priors by \cite{Woetal15}.

The choice of $\pi_{\boldsymbol \gamma}(\n \alpha, \n \beta_{\boldsymbol \gamma}, \sigma)$ from an objective point of view has been an important research question since at least \cite{ZellSiow80}; see \citet{liang08} and \citet{Baetal11} for in-depth reviews. A substantial part of the literature has focused on priors that have the peculiarity of using a mixture of normal densities for $\n\beta_\gamma$, centered at zero and with a variance proportional to the information matrix, and standard non-informative priors for the common parameters $\n\alpha$ and $\sigma$:
\begin{align}
&\pi_{\boldsymbol 0}(\n \alpha, \sigma) =  \sigma^{-1} \label{prior0}\\
&\pi_{\boldsymbol \gamma}(\n \alpha, \n \beta_{\boldsymbol \gamma}, \sigma) = \sigma^{-1}\ \int_0^{+\infty} N(\beta_{\boldsymbol \gamma}\mid \n 0, g\ \sigma^2\ (\n V_{\boldsymbol \gamma} ^T\n V_{{\boldsymbol \gamma}})^{-1})\ h_{\boldsymbol \gamma}(g)\ d g\ , {\boldsymbol \gamma}\neq \n 0 \label{priorgamma}
\end{align}
where $\n V_{\boldsymbol \gamma}=(\n I_n-\n P_0) \n X_{\boldsymbol \gamma}$, $\n P_0=\n X_0 (\n X_0^T \n X_0)^{-1} \n X_0^T$.
This approach was named \emph{conventional} by \cite{BergerPericchi01} and \cite{BayGar07}, a term that we also adopt. Conventional priors have been successfully implemented by many authors, like \cite{FLS01} and \cite{liang08} and, more recently, \cite{Baetal11} have shown that this class of priors satisfies a number of 
desirable properties including several types of invariance, predictive matching and consistency \citep[see][for full details]{Baetal11}.

Conventional priors lead to Bayes factors that depend on readily available statistics, namely\begin{equation}\label{BF}
B_{\boldsymbol \gamma}=\mathcal{B}\left( \frac{\textrm{SSE}_{\boldsymbol \gamma}}{\textrm{SSE}_{\boldsymbol 0}},k_0,k_{\boldsymbol \gamma}+k_0\right)\,
\end{equation}
where $\textrm{SSE}_{\boldsymbol \gamma}$ and $\textrm{SSE}_{\boldsymbol 0}$ are the sum of squared errors under model $M_{\boldsymbol \gamma}$ and $M_{\boldsymbol 0}$, respectively and
\begin{equation}\label{bcal}
\mathcal{B}(q,\kappa_0,\kappa_1)=\int (1+q\,g)^{-(n-\kappa_0)/2}\,(1+g)^{(n-\kappa_1)/2}\, h_{\boldsymbol \gamma}(g)\ d g.
\end{equation}

With respect to $h_{\boldsymbol \gamma}$, our recommended choice is the robust prior in \cite{Baetal11} that corresponds to 
\begin{align}
h_{\boldsymbol \gamma}(g)= \frac{1}{2}\ \left(\frac{1+n}{k_{\boldsymbol \gamma}+k_0}\right)^{1/2}\ (g+1)^{-3/2}\ , \quad g> \frac{n+1}{k_{\boldsymbol \gamma}+k_0}-1
\end{align}
and leads to a Bayes factor that can be expressed in closed-form:
\begin{multline}
\mathcal{B}(q,\kappa_0,\kappa_1)=
\left(\frac{n+1}{\kappa_1}\right)^{(\kappa_1-\kappa_0)/2}\ 
\frac{q^{-(n-\kappa_0)/2}}{\kappa_1+1}\\
{}_2F_1\left[\frac{\kappa_1-\kappa_0+1}{2};\frac{n-\kappa_0}{2};
\frac{\kappa_1-\kappa_0+3}{2};
\frac{\kappa_1 (1-1/q)}{n+1}\right]
\end{multline}
where ${}_2F_1$ is the standard hypergeometric function \citep{AS64}. The numerical results presented in this paper are based on the robust prior but, quite importantly, the theoretical results equally apply to any prior in the class of conventional priors defined above. 

Relevant for the problem with factors is the underlying assumption in the conventional approach that the matrix $[\n X_0 \mid \n X_{\boldsymbol \gamma}]$ is of full column rank, hence guaranteeing the existence of the inverse matrix $(\n V_{\boldsymbol \gamma} ^T\n V_{{\boldsymbol \gamma}})^{-1}$. As we will see in the next section, this condition is usually not satisfied when we consider the inclusion of factors in the variable selection problem. 


\section{Factors in Variable Selection}\label{sec.factors}
The methodology that we propose is presented in Section~\ref{general} in the general setting. This proposal follows after a discussion that extends comments in \citet{Chi96} about different possibilities for handling factors in variable selection. These arise as a consequence of their special structure, formed by levels, that can be treated together or separately. 
The next two sections are devoted to this discussion, which is presented in the context of the one factor case for clarity of exposition.

\subsection{Initial considerations}
The frequentist textbook approach to handling factors, particularly when only one factor is present, is a two-step procedure: first, perform the $F$-test for the hypothesis that there is no difference between the groups defined by the levels of the factor. Next, if that hypothesis is rejected, the question of which are the groups that are different is addressed. Answering that question is usually dealt with via pairwise $t$-tests or by obtaining and interpreting confidence intervals for the individual level effects. Issues related to multiple comparisons emerge, and there are many possible corrections prescribed to compensate for that \citep[see][for example]{Hsu96}. 

A naive Bayes analog of the first step above would consist in comparing 
the posterior probability of the model without the factor (null model) with the posterior probability of the model that states that all levels of the factor are relevant to explain the response (full model). This approach, which addresses the question ``are all levels relevant?''\ is detailed in Section~\ref{Areall}, where we derive formulas for posterior probabilities based on conventional priors and that extend previous findings by \cite{BayGar07} in the underlying rank deficient problem. One particular inconvenient of this procedure is that if the number of levels is relatively large, the posterior probability of the full model will be highly penalized due to its complexity, hence potentially underestimating the importance of only a small number of levels of the factor explaining the response.

Further, what is by far less clear is how to perform the second step, i.e., how to identify which levels are important when the full model receives substantial evidence from the data. One may start questioning whether a more parsimonious model can be selected, one in which only some levels of the factor are included. This poses a coherence problem in terms of the prior probabilities of the models that we will be entertaining, which in a way is similar to the multiple comparison issues in the frequentist analysis. 
Another possibility, suggested by \cite{Chi96}, is to use the posterior distribution under the full model to decide which levels are important. This would implicitly obviate the model selection uncertainty (supposing the full model is certainly the true model) but could be argued to imply a double use of the data: selecting the model and then estimating the parameters.

Remarkably, the Bayesian paradigm allows us to address the problem from a perspective which is different from the two-step procedures outlined above. 
It follows by recognizing from the very beginning that the problem of interest is determining which levels of the factor are relevant to explain the response. This cannot be captured by a pair of models and requires a collection of models indexed by the active levels.
The question being addressed is hence ``is at least one level relevant?'' and it is considered in Section~\ref{Isatleast} in the one-factor case, and further extended to the general case in Section~\ref{general}. This, which is ultimately our recommended strategy, allows for treating multiplicity issues through the prior probabilities over the model space as is done in the variable selection scenario. Furthermore, we will argue that the inclusion probabilities of the levels (a summary of the posterior distribution on the model space) may be used to ascertain the importance of the levels. We highlight that this sub-product of our proposed approach removes the need for any second step and is obviously formally coherent. 

Although in principle this approach looks straightforward, we shall see in Section~\ref{Isatleast} that there are a number of difficulties associated with its correct implementation. These have mainly to do with the assignment of the probabilities over the model space and the role of parametrizations. To the best of our knowledge none of them have been formally treated before in the literature. 

\subsection{The one-factor case}\label{1f}
In this section, we consider the situation where we have a number of sure variables but are uncertain about whether a factor $A$ should also be used to explain the response. We suppose that $A$ has $\ell$ possible levels. For illustrative purposes, of the running example introduced in Section~\ref{example}, we use $\{1,x_{01},x_{02},x_{03}\}$ as sure variables and for $A$ we use the factor Sports. Recall that this factor has $\ell=6$ levels. 

%

\subsubsection{Are all levels relevant?}\label{Areall}
When considering the problem of whether all levels are relevant, there are only two models to entertain: the model that contains all levels of $A$ ($M_{\boldsymbol 1}$), and the one without the factor ($M_{\boldsymbol 0}$). 

In its original form, model $M_{\boldsymbol 1}$ can be expressed as 
\begin{equation}\label{ANCOVA1f}
M_{\boldsymbol 1}:\ y_{ij}=\n x_{0ij}^T\n\alpha+a_j+\varepsilon_{ij},\ j=1,\ldots,\ell,\ i=1,\ldots,n_j.
\end{equation}
where $y_{ij}$ is the value of the response of the $i$-th individual in level $j$ of the factor. For purposes that will be clear in the sequel, this model is presented in its natural formulation, although it's a rank-deficient parametrization. 
In \eqref{ANCOVA1f}, the vector $\n x_{0ij}$ (of dimension $k_0$) contains the values of the sure variables (including at least the constant) and $a_j$ stands for the effect of the $j$-th level of the factor. A more compact expression for $M_{\boldsymbol 1}$ is
\begin{equation}\label{ANCOVA1fcomp}
M_{\boldsymbol 1}:\ \n y\mid \n\alpha, \n a, \sigma \sim N(\n X_0\n\alpha+\n X\n a,\sigma^2\n I_n)\ ,
\end{equation}
where $n=\sum_{j=1}^\ell n_j$, $\n y^T=(y_{11},\ldots,y_{1n_1},\ldots,y_{\ell 1},\ldots,y_{\ell n_\ell})$, $\n X=\oplus_{j=1}^\ell\n 1_{n_j}$, with $\oplus$ standing for direct sum of matrices.

Also
$$
M_{\boldsymbol 0}:\ y_{ij}=\n x_{0ij}^T\n\alpha+\varepsilon_{ij},\ j=1,\ldots,\ell,\ i=1,\ldots,n_j\ ,
$$
which can be rewritten as
$$
M_{\boldsymbol 0}:\ \n y\mid \n\alpha, \sigma \sim N(\n X_0\n\alpha,\sigma^2\n I_n)\ .
$$

For the model parameters under $M_{\boldsymbol 1}$, the prior in \eqref{priorgamma} cannot be used since the matrix $[\n X_0 \mid \n X]$ is not of full rank (it has $k_0+\ell$ columns, but its rank is $k_0+\ell-1$ since $\n X_0$ contains a vector of ones), which in turn implies that the inverse of $\n V^T \n V$ in that formula does not exist. 

Testing problems in rank-deficient settings were studied by \cite{BayGar07}. Their main practical conclusion is that, when computing the Bayes factor in \eqref{BF}, the third argument of ${\cal B}$, which corresponds to the number of columns in $[\n X_0 \mid \n X]$ in the full-rank case, should be replaced by the rank of $[\n X_0\mid \n X]$, that is,
\begin{equation}\label{BFnaive}
B_{\boldsymbol 1}=\mathcal{B}\left(\frac{\textrm{SSE}_{\boldsymbol 1}}{\textrm{SSE}_{\boldsymbol 0}},k_0,k_0+\ell-1\right).
\end{equation}
The theory in \cite{BayGar07} is developed under a condition 
of testability which, unfortunately, does not hold in general in this setting (this is part of Theorem~\ref{giBF} below). Nevertheless, there is a quite solid reason to think that \eqref{BFnaive} is still the right way to compare $M_{\boldsymbol 0}$ to model $M_{\boldsymbol 1}$: since sums of squared errors are invariant with respect to model reparametrizations, if we perform any full rank reparametrization of $M_{\boldsymbol 1}$, and subsequently apply \eqref{BF} (now that we have a full rank model), we would end up with \eqref{BFnaive}.

Despite all evidence in favor of \eqref{BFnaive} to compare model $M_{\boldsymbol 1}$ against model $M_{\boldsymbol 0}$, its justification as a Bayesian solution is at this point yet to establish: does it correspond to an actual Bayes factor for the competing models, arising from valid priors? The need for such requirement, to which we superscribe, was established as a Principle in \cite{BergerPericchi01}. 

What we formalize in the next two theorems is a positive answer to this requirement. We show in Theorem~\ref{giBF} that \eqref{BFnaive} results from using the prior in \eqref{priorgamma} with 
$(\n V^T \n V)^{-1}$ replaced with any element of a certain class of matrices that, accordingly to Theorem~\ref{res}, are all (non-singular) generalized inverses of $\n V^T \n V$. A related use of this type of generalized inverses has been considered by \citet{Beretal16} in a different problem, namely variable selection in regression problems with more predictors than data points.

Theorem~\ref{giBF} is presented in a more general setting, as this will be useful in the sequel.

\begin{thm}\label{giBF} Consider
$$
M_0:\ \n y\mid \n\alpha, \sigma \sim N(\n X_0\n\alpha,\sigma^2\n I_n),
$$
and the rank-deficient model
$$
M_A:\ \n y\mid \n\alpha, \n a, \sigma \sim N(\n X_0\n\alpha+\n X\n a,\sigma^2\n I_n)\ ,
$$
where $\n X$ is $n\times \ell$, $\rank(\n X_0)=k_0$, $\rank[\n X_0 \mid \n X]=k_0+r$ and $k_0<k_0+r<k_0+\ell$. Then, the hypothesis $H_0:\ \n a=\n 0$ which determines model $M_0$ from model $M_A$ will not, in general, be testable. 

Nevertheless, let the prior under $M_0$ be $\pi_0(\n\alpha,\sigma)=\sigma^{-1}$ and under $M_A$ be
\begin{equation}\label{eq.priorG}
\pi_A(\n\alpha,\n a,\sigma)=\sigma^{-1}\ \int  N(\n a\mid\n 0,g\sigma^2\n S)\, h_A(g)\, dg,
\end{equation}
where, with $\n V=(\n I-\n P_0)\n X$, 
\begin{equation}\label{eq.S}
\n S=(\n V^T\n V+\n T)^{-1}
\end{equation}
and $\n T$ is any symmetric semi-positive definite matrix of dimension $\ell \times \ell$ and rank $\ell-r$ such that $\n S$ exists. Then, the Bayes factor of $M_A$ to $M_0$ is given by
$$
{\cal B}_A\left(\frac{\textrm{SSE}_A}{\textrm{SSE}_0}, k_0, k_0+r\right),
$$
where ${\cal B}_A$ is the integral in \eqref{bcal} with $h_\gamma$ replaced by $h_A$. Note that $B_A$ does not depend on the particular choice of $\n T$.
\end{thm}
\begin{proof}
  See Appendix~\ref{pgiBF}.
\end{proof}

Hence, the proposed Bayes factor results from a proper prior on $\n a$, in agreement with the mentioned Principle in \cite{BergerPericchi01} and further with the first criterion (called Basic) in \cite{Baetal11}. Additionally, the matrix \eqref{eq.S} is a generalized inverse of $\n V^T\n V$, as we state in the following result. Both observations confirm the proposed prior as a generalization of \eqref{priorgamma} to the setting of rank-deficient models.

\begin{thm}\label{res}
The matrix (\ref{eq.S}) is a generalized non-singular inverse of $\n V^T\n V$.
\end{thm}
\begin{proof}
  See Appendix~\ref{pres}.
\end{proof}

At this point, Theorem~\ref{giBF} can be seen as a minor formal technicality. Nevertheless, it makes it possible to handle directly the original rank-deficient parametrization, a usage that will turn out crucial in the next section, establishing \eqref{BFnaive} the basis for our proposed conventional solution to handle factors. 

The fact that we use generalized inverses as a vehicle to work directly with \eqref{ANCOVA1f} makes our approach different from the proposal in \cite{Rouetal12} and \cite{Rouetal17}, where a full rank parametrization is \emph{chosen} and independent normal priors are used for the new coefficients of the resulting models. Regardless of how sensible is the chosen parametrization (which we think is sensible) this choice has an impact. This observation will also be revisited in the next section. 

In our running example, we would obtain a Bayes factor of $M_{\boldsymbol 1}$ to $M_{\boldsymbol 0}$ of $B_{\boldsymbol 1}=4$, which implies substantial support to $M_{\boldsymbol 1}$ (all levels are needed) in detriment of $M_{\boldsymbol 0}$.

\subsubsection{Is at least one level relevant?}\label{Isatleast}
The notion that a factor explains $y$ when at least one of its level is relevant cannot be described by a single model and requires a collection of models: one per each of the possible combination of levels. In particular, and making use of the notation in \eqref{Mgamma}, this hypothesis holds true if and only if any of the models
$$
M_{\boldsymbol \gamma}:\ \n y\mid \n\alpha,\n a_{\boldsymbol \gamma}, \sigma \sim N( \n X_0\n\alpha+\n X_{\boldsymbol \gamma}\n a_{\boldsymbol \gamma},\sigma^2\n I_n) 
$$
is the true model, where ${\boldsymbol \gamma}\in\{0,1\}^\ell\backslash \{\n 0\}$. Here, as in Section~\ref{sec.vsrm}, matrix $\n X_{\boldsymbol \gamma}$ represents the submatrix of $\n X$ in \eqref{ANCOVA1fcomp} that results from selecting the columns that correspond to ones in ${\boldsymbol \gamma}$, and hence there are $2^\ell-1$ such models. This of course reminds us of  the standard variable selection problem, but here the probability that the factor $A$ is relevant should be obtained as
\begin{equation}\label{probF}
P(A\mid\n y)=\sum_{{\boldsymbol \gamma}\in\{0,1\}^\ell:\ {\boldsymbol \gamma}\ne \boldsymbol{0}}\, P(M_{\boldsymbol \gamma}\mid \n y)=1-P(M_{\boldsymbol 0}\mid \n y)\,
\end{equation}
that can be written as a function of the Bayes factor of $M_{\boldsymbol \gamma}$ to $M_{\boldsymbol 0}$ and of the prior probabilities of the models, for which we recommend the conventional priors \eqref{priorgamma} with $(\n V^T\n V)^{-1}$ replaced by the regular generalized inverse \eqref{eq.S} when $M_{\boldsymbol \gamma}$ is rank deficient [there is only one model which is rank deficient ($M_{\boldsymbol 1}$) and the rest are all full rank]. Hence, in practical terms, when computing the Bayes factors, $B_{\boldsymbol 1}$ will have the expression \eqref{BFnaive} while all the other $B_{\boldsymbol \gamma}$ follow the standard expression 
$\mathcal{B}(\textrm{SSE}_{\boldsymbol \gamma}/\textrm{SSE}_{\boldsymbol 0},k_0,k_0+k_{\boldsymbol \gamma})$, 
where $k_{\boldsymbol \gamma}$ is the number of levels that are active in $M_{\boldsymbol \gamma}$, that is, $\n 1^T {\boldsymbol \gamma}$. Notice that the models with $k_{\boldsymbol \gamma}=\ell-1$ are all full rank parametrizations of $M_{\boldsymbol 1}$, so that the Bayes factor for these coincides with $B_{\boldsymbol 1}$.

\paragraph{Prior probabilities} One very important ingredient in \eqref{probF} are the prior model probabilities $P(M_{\boldsymbol \gamma})$.

In principle, we could use any of the standard choices like \eqref{eq.const} or \eqref{eq.SB} with $k$ replaced by $\ell$ \citep[the constant prior \eqref{eq.SB} was used by][in their application]{{ClyPar98}}. Any of these possibilities ignores the common nature of the $M_{\boldsymbol \gamma}$ models and have also the undesirable property of apportioning the probabilities in a way that strongly depends on the number of levels of $A$ (which in many cases has associated a certain degree of arbitrariness). In the case of the constant prior, this dependency has a large impact on the prior probabilities, particularly visible in reducing the probability of $M_{\boldsymbol 0}$ as $\ell$ increases. Hence, for instance, if $\ell=4$ then $P(M_{\boldsymbol 0})=1/16$ while if $\ell=6$ then $P(M_{\boldsymbol 0})=1/64$. The case of the Scott-Berger prior is less dramatic, but still has a non-negligible effect: $P(M_{\boldsymbol 0})=1/5$ for $\ell=4$ and $P(M_{\boldsymbol 0})=1/7$ for $\ell=6$.

A different possibility, which is the one that we recommend, is to recognize the hierarchical nature of the testing problem and first elicit $P(M_{\boldsymbol 0})=P(A)=1/2$, and then use one of the expressions \eqref{eq.const} or \eqref{eq.SB} to determine the conditionals $P(M_{\boldsymbol \gamma}\mid A)$. Our preferred option is the \cite{ScottBerger09} prior in \eqref{eq.SB} because it automatically controls for the multiplicity issue that arises due to the $\ell$ dummy variables used. Remarkably, this potential pitfall was observed by \cite{Chi96} 
leading to the recommendation of handling factors in blocks, in a strategy similar to that seen in the previous section. Our proposed prior is then:
\begin{equation}\label{SB}
P(M_{\boldsymbol \gamma}\mid A)=\frac{1}{\ell {\ell \choose k_{\boldsymbol \gamma}}}\ .
\end{equation}
Obviously, with this hierarchical approach the probability of $M_{\boldsymbol 0}$ remains unchanged with the number of levels. Moreover, this proposal is in agreement with the ``effect hierarchy'' principle  (\ \cite{BinChi07}; \cite{ConDel16}), as models with the same number of active levels  will have the same prior probability, and the higher the number of active levels, the smaller is the prior probability of a model.

In our example, we have computed the posterior probability of the factor being relevant using \eqref{probF}, and it resulted in $P(A\mid\n y)=0.997$. This implies very strong evidence supporting the conjecture that sports activity explains body mass index. This was obtained using our preferred prior on the model space: $P(M_{\boldsymbol 0})=1-P(A)=1/2$, followed by \eqref{SB}. It is possible to obtain the inclusion probabilities of the levels of the factors, and its role in understanding the effect of each of the levels is considered in the general setting in Section~\ref{incprob}.

\paragraph{The role of reparametrization} In our setting, the model that nests all the competing models (the full model) is not full rank. Nevertheless, and remarkably, the statistical analysis is based on its original formulation. Reparametrizations takes place exclusively for mathematical reasons --- namely achieving a full rank representation of said model --- and that is simply not needed here.

Under the setup of Section~\ref{Areall}, we noticed that this wasn't really relevant, as any full rank expression of the model would give rise to the same Bayes factor.
This could give us the illusory perception that the way the model is initially parametrized does not have any impact on the results when we approach the question computing $P(A\mid\n y)$ as in here, but (quite surprisingly) this turns out to be wrong. In our example, if we adopt the hierarchical Scott and Berger prior (similar results are obtained with the other priors) and we parametrize \eqref{ANCOVA1f} using the first level as the baseline, then we obtain $P(A\mid\n y)=0.560$. If, on the other hand, we use the second level as the baseline, we obtain $P(A\mid\n y)=0.998$. This happens because in this type of reparametrization the effect of the baseline level is included in the null. Since level 1 is quite important, $P(A\mid \n y)$ is more or less large depending on whether level 1 is chosen as the baseline. 

This is very unsatisfactory, and should come as a warning: if one chooses to reparametrize \eqref{ANCOVA1f} to obtain a full rank model, $P(A\mid\n y)$ --- i.e., the posterior probability that any of the levels of the factor is  relevant in explaining the response --- will depend on the parametrization chosen. This is even more worrisome as the choice of reparametrization is in many occasions arbitrary (even made by the statistical software used) and the practitioner will not be in general aware of its consequences.

One can come up with ways of parametrizing that are more satisfactory \citep[e.g. the one in][seems to us very reasonable]{Rouetal12}, but we should be aware that results are dependent on that choice (and we can imagine many sensible ways of reparametrizing). 

This is one of the main strengths of our methodology: it works directly with the natural formulation of the model, the rank deficient specification, and hence it does not depend on any kind of full rank reparametrization.

\paragraph{The two-levels case} The unanimous way of handling factors with only two levels is by coding them using a 0-1 variable. Such approach implicitly implies choosing one of the levels (that coded with a zero) as the baseline. In this situation, the two models being entertained are (supposing that the first level is the baseline)
\begin{align*}
M_0^\star:\ y_{ij}&=\n x_{0ij}^T\n\alpha+\varepsilon_{ij}\\
M_1^\star:\ y_{i1}&=\n x_{0i1}^T\n\alpha+\varepsilon_{i1}\ y_{i2}=\n x_{0i2}^T\n\alpha+\delta+\varepsilon_{i2}\ .
\end{align*}
The effect of the factor is measured by $P(M_1^\star\mid\n y)$, and it is easy to see that, in this case, this does not depend on the choice of the baseline. The obvious objective specification for the prior probabilities is $P(M_0^\star)=1/2$.

Our approach is in principle different, as we propose handling the original rank-deficient formulation in \eqref{ANCOVA1f}. Here, four models are entertained: 
\begin{align*}
M_{\boldsymbol 0}:\ y_{ij}&=\n x_{0ij}^T\n\alpha+\varepsilon_{ij}\\
M_{(0,1)}:\ y_{i1}&=\n x_{0i1}^T\n\alpha+\varepsilon_{i1},\,\,y_{i2}=\n x_{0i2}^T\n\alpha+a_2+\varepsilon_{i2}\\
M_{(1,0)}:\ y_{i1}&=\n x_{0i1}^T\n\alpha+a_1+\varepsilon_{i1},\,\,y_{i2}=\n x_{0i2}^T\n\alpha+\varepsilon_{i2}\\
M_{\boldsymbol 1}:\ y_{i1}&=\n x_{0i1}^T\n\alpha+a_1+\varepsilon_{i1},\,\,y_{i2}=\n x_{0i2}^T\n\alpha+a_2+\varepsilon_{i2}\ .
\end{align*}
The probability of the factor being relevant in explaining the response is
$$
P(A\mid\n y)=P(M_{(0,1)}\mid\n y)+P(M_{(1,0)}\mid\n y)+P(M_{\boldsymbol 1}\mid\n y)\ .
$$
It can be easily checked that ${\cal B}_{(0,1)}={\cal B}_{(1,0)}={\cal B}_{\boldsymbol 1}$ which at the same time is equal to the Bayes factor of $M_1^\star$ to $M_0^\star$ above. Hence, if the prior probabilities are assigned hierarchically and $P(M_{\boldsymbol 0})=1/2$, then $P(A\mid\n y)=P(M_1^\star\mid\n y)$ agreeing with intuition. We take this as added support for the hierarchical specification of the prior probabilities over the model space.
There is no similar coincidence when the number of levels in the factor is greater than 2.

\subsection{The general case}\label{general}
In the general case, we have $p$ factors (factor $A_r$ has $\ell_r$ levels, each with a coefficient $a_{rj}$), $k$ variables and $k_0$ sure variables. The full model is hence
\newcommand{\snj}{\text{\bfseries\itshape j}} 
\begin{multline}\label{ANCOVAgf}
M_{\boldsymbol 1}:\ y_{i\snj}=\n x_{0i\snj}^T\n\alpha+\n x_{i\snj}^T\n\beta+
a_{1j_1}+a_{2j_2}+\cdots+a_{pj_p}+\varepsilon_{i\text{\bfseries\itshape j}}\, ,\\
j_r=1,\ldots,\ell_r,\, i=1,\ldots,n_\snj\ , r=1, \ldots, p,
\end{multline}
with $\n j$ representing the vector of indexes $(j_1, \ldots, j_p)$. The sample size is then $n=\sum_\snj n_\snj$ and we assume that there is at least one observation per group, i.e., $n_\snj\ge 1$ for all $\n j$. In matrix notation, $M_{\boldsymbol 1}$ can be expressed as
\begin{equation}\label{ANCOVAgfcomp}
M_{\boldsymbol 1}:\ \n y\mid \n\alpha, \n \beta, \n a, \sigma^2 \sim N(\n y\mid \n X_0\n\alpha+\n X\n\beta+\n Z\n a, \sigma^2\n I_n),
\end{equation}
where $\n X_0$ is $n\times k_0$; $\n X$ is $n\times k$ and $\n Z$ is $n\times L$, where $L=\sum_{r=1}^p\, \ell_r$. The design matrix in $M_{\boldsymbol 1}$ is then $[\n X_0 \mid \n X \mid \n Z]$ and is rank-deficient since its rank is $k_0+k+L-p$.

As before, the null model is
$$
M_{\boldsymbol 0}:\ y_{i\snj}=\n x_{0i\snj}^T \n\alpha+\varepsilon_{i\text{\bfseries\itshape j}},\,\, j_r=1,\ldots,\ell_r,\,i=1,\ldots,n_\snj\ .
$$


In this general setting, there are a total of $2^{k+L}-1$ possible models $M_{\boldsymbol \gamma}$ that are nested in $M_{\boldsymbol 1}$ and contain the sure variables. In our example, we have $k=2$, $L=6+3=9$ so that there are a total of $2^9=512$ competing models.
The posterior probability of any of these models is proportional to $B_{\boldsymbol \gamma}\, P(M_{\boldsymbol \gamma})$, and, as argued in the previous section, no matter if $M_{\boldsymbol \gamma}$ is full rank or rank-deficient, $B_{\boldsymbol \gamma}$ should be obtained as
\begin{equation}
B_\gamma=\mathcal{B}\left(\frac{\textrm{SSE}_{\boldsymbol \gamma}}{\textrm{SSE}_{\boldsymbol 0}},k_0,r_{\boldsymbol \gamma}\right),
\end{equation}
where $r_{\boldsymbol \gamma}$ is the rank of the design matrix in $M_{\boldsymbol \gamma}$. 

What we have to discuss now is the prior on the model space. In this general case, the number of models increases very fast with either $p$ or any $\ell_r$, therefore amplifying the effect of the standard choices of priors over e.g. the probability of the null that we observed in the case with only one factor.

Extending the previous reasoning, our proposal is that prior probabilities must be assigned hierarchically. Initially, the probability that a certain number, say, $m_1+m_2$, of variables and factors in $\{x_1, \ldots, x_k, A_1, \ldots, A_p\}$ are relevant to explain the response is established, and then, conditionally on this, the probability of individual models in this category is specified:
\begin{align}
&P(\{x_{i_1},\ldots,x_{i_{m_1}},A_{j_1},\ldots,A_{j_{m_2}}\})\label{p1} \\ 
&P(M_{\boldsymbol \gamma}\mid \{x_{i_1},\ldots,x_{i_{m_1}},A_{j_1},\ldots,A_{j_{m_2}}\})\ . \label{p2}
\end{align}
In \eqref{p2}, we are assuming that $\n\gamma$ is compatible with the given configuration of predictors; otherwise, that probability is zero.

There are various possibilities to determine these probabilities, but as already argued  our preferred option for both stages is to use the \cite{ScottBerger09} prior, as this choice controls for multiplicity both in the usual way (over the total number of predictors) but also over the number of levels of the factors. Straightforward combinatorics arguments lead to the following expressions:
\begin{align}
&P(\{x_{i_1},\ldots,x_{i_{m_1}},A_{j_1},\ldots,A_{j_{m_2}}\})=\left[(k+p+1){k+p\choose m_1+m_2}\right]^{-1}\label{p3}\\ 
&P(M_{\boldsymbol \gamma}\mid \{x_{i_1},\ldots,x_{i_{m_1}},A_{j_1},\ldots,A_{j_{m_2}}\})=\left[\prod_{h=1}^{m_2}\, \ell_h {\ell_h\choose k_\gamma^h}\right]^{-1}\ ,\label{p4}
\end{align}
where, in \eqref{p4}, $m_2\ge 1$ (otherwise, it is equal to one), and $1\le k_\gamma^h\le \ell_h$ is the number of levels of factor $A_h$ active in $M_{\boldsymbol \gamma}$.

Since \eqref{p3} only depends on the number of predictors, it's clear that models with the same number of predictors will be assigned the same marginal probability. Additionally, it's easy to verify that the marginal probability of a factor $A_r$ (i.e. the sum of the prior probabilities of the $2^{k+L-\ell_r}(2^{\ell_r}-1)$ models that contain at least one of the levels of factor $A_r$) is 1/2, and this (along with \eqref{p4}) shows that the present strategy extends the reasoning of Section~\ref{Isatleast}. Similarly, the marginal probability that each of the variables is included \textit{a priori} is also 1/2. This is again all in agreement with the ``effect hierarchy'' principle of \cite{BinChi07} \cite[c.f.\ also][]{ConDel16}.

Having obtained the posterior model probabilities, it is just a question of how to summarize them to be able to provide measures of the evidence that any of the factors is relevant in explaining the response. One obvious possibility is the analogous to the posterior inclusion probabilities in the standard variable selection problem. To obtain the inclusion probability of a factor $A_r$ it suffices to sum the posterior probabilities of all the models that contain at least one of the levels of factor $A_r$. The inclusion probabilities for the variables can be obtained as usual. 

Regarding the obesity example, in Table~\ref{incprobfactors} we have collected the posterior inclusion probabilities of all factors and variables. The conclusion is straightforward and states that both of the factors and $x_1$ are very relevant in explaining the body mass index while the evidence about the role of $x_2$ (hours of sleep) is not conclusive.

\begin{table}[t!]
\begin{center}
{\small\scalebox{0.75}{
\begin{tabular}{cccc}
$A_1$ & $A_2$  & $x_1$  & $x_2$\\
\hline
0.995 & 0.998 &  0.999 &  0.622\\
\hline
\end{tabular}
}}
\end{center}
\caption{\small Inclusion probabilities of factors and variables. \label{incprobfactors}}
\end{table}

\subsection{The inclusion probabilities of the levels}\label{incprob}
The approach we have introduced allows us to measure the importance of the individual levels of each of the factors by examining the associated inclusion probabilities, i..e, the sum of the posterior probabilities of all the models in which that level appears, and this is a distinctive feature of the methodology here proposed.  

While the inclusion probabilities cannot tell us the direction of the effect of the levels nor its magnitude (as any other product of a model selection exercise) they can be used to ascertain which levels are relevant and which are not, hence implicitly suggesting which categories of the levels can be included in the overall mean.

We have computed the inclusion probabilities for the levels of the factors $A_1$ and $A_2$ for the childhood obesity example. These are collected in Table~\ref{inclproblev}. Our interpretation of these results is as follows. The fact that the inclusion probability of level 1 of the factor $A_1$ is high, and the other ones low, means that, averaging out the effect of all other predictors, a child in level 1 has a mean body mass index which is deemed as significantly different from the overall mean. On the other hand, a child in the other levels will have a mean body mass index which is identical to the overall mean. 
When we look at the posterior inclusion probabilities of the levels of factor $A_2$, these are all relatively large, which makes their interpretation not as straightforward as before. What we could say is that all levels are relevant  in the sense that the mean body mass index of a child in any of the three levels will be significantly different from the overall mean. What is not clear is, for instance, whether the effect of level 2 is similar to that of level 1.

\begin{table}[t!]
\begin{center}
{\small\scalebox{0.75}{
\begin{tabular}{cccccc|ccc}
\multicolumn{6}{c}{$A_1$} & \multicolumn{3}{c}{$A_2$}\\
1 & 2 & 3 & 4 & 5 & 6 & 1 & 2 & 3\\
0.99 & 0.08 & 0.25 & 0.09 & 0.14 & 0.09 & 0.82 & 0.76 & 0.78\\
\hline
\end{tabular}
}}
\end{center}
\caption{\small Inclusion probabilities of levels of factors. \label{inclproblev}}
\end{table}

\section{Discussion}
In Bayesian variable selection problems, considering factors in the list of potential predictors 
creates certain peculiarities that need to be carefully addressed. These have to do with the choice of a full rank representation of the underlying model, and the prior distributions over the model space and on the model-specific parameters.

We have developed methodology that handles all these issues generalizing the use of conventional priors \citep[a class that satisfies a number of optimal properties as seen in][]{Baetal11} and the \cite{ScottBerger09} prior that controls for multiplicity. The end result is a fully automatic procedure for variable selection, requiring no tuning from the user when it comes to the priors used, but also no need to decide on any type of full rank parametrization. An interesting subproduct of our proposal are the inclusion probabilities of the levels of a factor. We have argued that these contain valuable evidence to ascertain the individual contribution of the levels hence eliminating the need of any ulterior analysis. 

We have not considered interaction terms and have not made any distinction between ordered and unordered factors. The former limitation clearly requires more research and will be pursued elsewhere. When it comes to the latter, it is not clear to us how an objective approach to variable selection in this context can take this information into account.

\appendix
\section{Proofs}
%

\subsection{Proof of Theorem ~\ref{giBF}}\label{pgiBF}
First, we state the following preliminary result. 

\begin{lem}\label{lemR}
Let $\n X_0$ be a $n\times k_0$ matrix, and $\n X$ is $n\times k$. If $\rank(\n X_0)=k_0$ and $\rank[\n X_0\mid \n X]=r+k_0$ then $\rank((\n I-\n P_0)\n X)=r$.
\end{lem}
\begin{proof}
Here we adopt the notation in \cite{Harv:1997}: $\cal{C}(\cdot)$ denotes the space spanned by the columns of a matrix while $\cal{N}(\cdot)$ stands for the null space of a matrix.

With the assumptions in the result, $\n X$ can be expressed as $[\n X_c \mid \n X_d]$ where the first $r$ columns (defining $\n X_c$) do not belong to the space ${\cal C}(\n X_0)$ and the remaining $k-k_0$ do. Hence $\rank((\n I-\n P_0)\n X)=\rank((\n I-\n P_0)\n X_c)$. Now, notice that ${\cal C}(\n X_c)$ and ${\cal N}(\n I-\n P_0)$ are disjoint since ${\cal N}(\n I-\n P_0)\subset {\cal C}(\n P_0)={\cal C}(\n X_0)$. Then, and due to Theorem 17.5.4 in \cite{Harv:1997}, $\rank((\n I-\n P_0)\n X_c)=\rank(\n X_c)=r$ which completes the proof.
\end{proof}

And now we prove Theorem~\ref{giBF}.
\begin{proof}

In the context of a linear model $\n y = \n Z \n \eta+\n \varepsilon$, with $E[\n y \mid \n Z] = \n Z \n\eta$, an hypothesis $H:\n L\n \eta = \n 0$ is said to be testable if $\n L \n \eta$ is estimable, that is, if $\n L = \n B\n Z$ for some $\n B$, or equivalently if $\n L$ is the row space of $\n Z$. Consider the case where we have a single factor with 2 levels, and two observations have been taken at each of the levels, so that $n=4$, $\ell=3$. Additionally, assume that $X_0=\n 1$. Hence, $\n Z= [\n 1_4 | \oplus_{i=1}^2 \n 1_2]$, with $\n 1_p$ representing a $p$-dimensional vector of ones, so that $k_0=r=1$. Any 1-dimensional estimable function of $\n \eta=(\alpha, a_1, a_2)^T$ must be written in the form
$$
\n L \n \eta = (L_1+L_2)\alpha + L_1 a_1 +L_2 a_2
$$
for any values of $L_1$, $L_2$ and $L_3$. As a consequence, neither $a_1$ nor $a_2$ are estimable, and hence $H_0:\ a_1=a_2=0$ is not testable.

We have that the Bayes factor is $\int (m_A(\n y\mid g)/m_0(\n y))h_A(g)\, dg$, where
\begin{eqnarray*}
m_A(\n y\mid g)&=&\int \sigma^{-1}\ N(\n y\mid \n X_0\n\alpha+\n X\n a,\sigma^2\n I_n)\ N(\n a\mid\n 0,g\sigma^2\n S)\
d\sigma\, d\n a\, d\n\alpha\\
&=& \int \sigma^{-1}\ N(\n y \mid\n X_0\n\alpha+\n V\n a,\sigma^2\n I_n)\ N(\n a\mid\n 0,g\sigma^2\n S)\
d\sigma\, d\n a\, d\n\alpha\, ,
\end{eqnarray*}
where the equality holds since the change of variable has a unit Jacobian. Now $\rank(\n V)=\rank(\n X^T(\n I-\n P_0)\n X)=\rank((\n I-\n P_0)\n X)=r$ by the result in Lemma~\ref{lemR}. Hence, the spectral decomposition of $\n V^T\n V$ is
$$
\n Q^T\n V^T\n V\n Q=\left(
\begin{array}{cc}
\n D & \n 0\\
\n 0 & \n 0
\end{array}
\right),
$$
where $\n D$ is diagonal, has dimension $r\times r$ and has positive entries and $\n Q$ is orthogonal. Now consider $\n Q$ partitioned as $\n Q=[\n Q_1\mid\n Q_2]$ where $\n Q_1$ is $k\times r$. Note that $\n V\n Q_2=\n 0$ and $\n V=\n V\n Q_1\n Q_1^T=\n L\n Q_1^T$ where $\n L=\n V\n Q_1$. 

Since $\n T$ is symmetric of rank $\ell-r$ it can be factorized (use a full rank factorization) as $\n T=\n C^T\n C$ where $\n C$ is $(\ell-r)\times \ell$ and has rank $\ell-r$. Notice the following equivalence of determinants
\begin{eqnarray*}
|\n S|^{-1}&=&|\n V^T\n V+\n C^T\n C|=|\n V^T(\n I-\n P_0)\n V+\n C^T\n C|=\\
&=&
\Big|
\Big(
\begin{array}{c}
(\n I-\n P_0)\n V \\ \n C
\end{array}
\Big)^T
\Big(
\begin{array}{c}
(\n I-\n P_0)\n V \\ \n C
\end{array}
\Big)
\Big|=\\
&=&
\Big|
\Big(
\begin{array}{c}
(\n I-\n P_0)\n L\n Q_1^T \\ \n C
\end{array}
\Big)^T
\Big(
\begin{array}{c}
(\n I-\n P_0)\n L\n Q_1^T \\ \n C
\end{array}
\Big)
\Big|=\\
&=&
\Big|
\Big[
\Big(
\begin{array}{cc}
(\n I-\n P_0)\n L & \n 0\\
\n 0 &\n I
\end{array}
\Big)
\Big(
\begin{array}{c}
\n Q_1^T \\ \n C
\end{array}
\Big)
\Big]^T
\Big[
\Big(
\begin{array}{cc}
(\n I-\n P_0)\n L & \n 0\\
\n 0 & \n I
\end{array}
\Big)
\Big(
\begin{array}{c}
\n Q_1^T \\ \n C
\end{array}
\Big)
\Big]
\Big|
=\\
&=&
\Big|
\Big(
\begin{array}{c}
\n Q_1^T \\ \n C
\end{array}
\Big)^T
\Big(
\begin{array}{cc}
\n L^T(\n I-\n P_0)\n L & \n 0\\
\n 0 &\n I
\end{array}
\Big)
\Big(
\begin{array}{c}
\n Q_1^T \\ \n C
\end{array}
\Big)
\Big|=\\
&=&\Big|
\Big(
\begin{array}{c}
\n Q_1^T \\ \n C
\end{array}
\Big)
\Big|^2\,\,
|\n L^T(\n I-\n P_0)\n L|=
\Big|
\Big(
\begin{array}{c}
\n Q_1^T \\ \n C
\end{array}
\Big)
\Big|^2\,\,
|\n Q_1^T\n V^T\n V\n Q_1|=
\Big|
\Big(
\begin{array}{c}
\n Q_1^T \\ \n C
\end{array}
\Big)
\Big|^2\,\,
|\n D|
\end{eqnarray*}

In particular, the above shows that the squared matrix 
$
\big(
\begin{array}{c}
\n Q_1^T \\ \n C
\end{array}
\big)
$
is non-singular.

Now
\begin{eqnarray*}
m_1(\n y\mid g)&=&\int \sigma^{-1}\ N(\n y\mid\n X_0\n\alpha+\n V\n a,\sigma^2\n I_n)\ \big(\sigma\sqrt{2\pi g}\big)^{-\ell}|\n S|^{-1/2}\,\times\\
&&\times  \exp\left\{\frac{1}{2\sigma^2 g}[\n a^T\n V^T(\n I-\n P_0)\n V\n a+\n a^T\n C^T\n C\n a]\right\}\
d\sigma\, d\n a\, d\n\alpha\\
&=&\int \sigma^{-1}\ N(\n y\mid \n X_0\n\alpha+\n L\n Q_1^T\n a,\sigma^2\n I_n)\ \big(\sigma\sqrt{2\pi g}\big)^{-\ell}|\n S|^{-1/2}\,\times\\
&&\times \exp\left\{\frac{1}{2\sigma^2 g}[\n a^T\n Q_1\n L^T(\n I-\n P_0)\n L\n Q_1^T\n a+\n a^T\n C^T\n C\n a]\right\}
d\sigma\, d\n a\, d\n \alpha.
\end{eqnarray*}
In the integral above, make the change of variables $\n a_1=\n Q_1^T\n a$ and $\n a_2=\n C\n a$ (note that $\n a_1\in\erre^r$ and $\n a_2\in\erre^{\ell-r}$) with associated Jacobian 
$
{\cal J}=
\big|
\big(
\begin{array}{c}
\n Q_1^T \\ \n C
\end{array}
\big)
\big|^{-1}
$
to obtain
\begin{eqnarray*}
m_A(\n y\mid  g)&=&\int \sigma^{-1}\ N(\n y\mid \n X_0\n\alpha+\n L\n a_1,\sigma^2\n I_n)\big(\sigma\sqrt{2\pi g}\big)^{-r}|\n S|^{-1/2}\,\times\\
&&\times\,  \big(\sigma\sqrt{2\pi g}\big)^{-(\ell-r)}
\exp\left\{-\frac{1}{2\sigma^2 g}(\n a_1^T\n L^T(\n I-\n P_0)\n L \n a_1\right\}\, \times \\
&& \times\, \exp\left\{-\frac{1}{2\sigma^2 g}\n a_2^T\n a_2\right\}
{\cal J}
d\sigma\, d\n a_1\, d\n a_2\, d\n\alpha\ .
\end{eqnarray*}
Now integrate out $\n a_2$ to obtain
\begin{eqnarray*}
m_A(\n y\mid g)&=& \int \sigma^{-1}\ N(\n y \mid \n X_0\n\alpha+\n L\n a_1,\sigma^2\n I_n)\ \big(\sigma\sqrt{2\pi g}\big)^{-r}|\n S|^{-1/2}\,\times\\
&&\times 
\exp\left\{-\frac{1}{2\sigma^2 g}(\n a_1^T\n L^T(\n I-\n P_0)\n L \n a_1\right\}
{\cal J}
d\sigma\, d\n a_1\, d\n\alpha\\
&=&
|\n S|^{-1/2}{\cal J}\; \left|\n L^T(\n I-\n P_0)\n L\right|^{-1/2}\times\\
&& \times\int \sigma^{-1}\ N(\n y\mid \n X_0\n\alpha+\n L\n a_1,\sigma^2\n I_n)\,\times\\ 
&&\qquad\qquad\times\, N(\n a_1\mid\n 0,g\sigma^2(\n L^T(\n I-\n P_0)\n L)^{-1})\ d\sigma\, d\n a_1\, d\n\alpha\\
&=&|\n S|^{-1/2}{\cal J}\, \left|\n L^T(\n I-\n P_0)\n L\right|^{-1/2}\times\\ 
&&m_0(\n y)\, \left(1+g\, \frac{\textrm{SSE}_A}{\textrm{SSE}_0}\right)^{-(n-k_0)/2}(1+g)^{(n-r-k_0)/2}
\ .
\end{eqnarray*}
The last equality is a basic one in conventional theory and can be found for example in \cite{BayGar07}. Finally, to complete the proof it suffices to show that 
$$
|\n S|^{-1/2}{\cal J}\; \left|\n L^T(\n I-\n P_0)\n L\right|^{-1/2}=1\ ,
$$
but this can be easily obtained from the equalities between determinants deduced above.

\end{proof}

\subsection{Proof of Theorem~\ref{res}}\label{pres}
\begin{proof}
From Theorem 18.2.5, page 421 in \cite{Harv:1997}, it suffices to show that, in the conditions of Theorem~\ref{giBF}, ${\cal C}(\n V^T\n V)$ and ${\cal C}(\n T)$ are essentially disjoint. Suppose these are not, and $\dim({\cal C}(\n V^T\n V) \cap {\cal C}(\n T))=d>0$. In this case:
\begin{eqnarray*}
\ell &=&\dim({\cal C}(\n V^T\n V) + {\cal C}(\n T))=\\
&=&\dim({\cal C}(\n V^T\n V)) + \dim({\cal C}(\n T))-\dim({\cal C}(\n V^T\n V) \cap {\cal C}(\n T))=\\
&=&r+(\ell-r)-d<\ell\, ,
\end{eqnarray*}
which proves that $d=0$. 
\end{proof}

\bibliography{./mibibliografia,./Anabel}

\bibliographystyle{imsart-nameyear}

\end{document}